\documentclass[%
    a4paper,
    12pt,
    DIV=12,
    BCOR=0mm,
    twoside=semi,
    headings=small
]{scrartcl}

\usepackage{scrpage2}

\addtokomafont{title}{\normalfont \bfseries}
\addtokomafont{section}{\normalfont \bfseries}
\addtokomafont{subsection}{\normalfont \bfseries}
\addtokomafont{subsubsection}{\normalfont \bfseries}
\addtokomafont{paragraph}{\normalfont \bfseries}
\newcommand{\quoteparagraph}[1]{\noindent{\normalsize \bfseries #1}\hspace{1em}}

\usepackage[utf8]{inputenc}           
\usepackage{amsmath, amsthm, amssymb} 
\usepackage[english]{babel}           
\usepackage{bibgerm}                  
\usepackage{graphicx}                 
\usepackage{hyperref}                 
\usepackage{booktabs}                 

\newcommand{\R}{\ensuremath{\mathbb{R}}}

\newcommand{\N}{\ensuremath{\mathbb{N}}}
\newcommand{\E}{\ensuremath{\mathbb{E}}}
\renewcommand{\P}{\ensuremath{\mathbb{P}}}

\newcommand{\cD}{\ensuremath{\mathcal{D}}}
\newcommand{\cU}{\ensuremath{\mathcal{U}}}
\newcommand{\cR}{\ensuremath{\mathcal{R}}}
\newcommand{\cZ}{\ensuremath{\mathcal{Z}}}

\newcommand{\cDD}{\ensuremath{\cD_{\textnormal{cur}}}}

\newcommand{\rk}[1]{\ensuremath{\left(#1\right)}}
\newcommand{\ek}[1]{\ensuremath{\left[#1\right]}}

\DeclareMathOperator{\TWR}{TWR}
\DeclareMathOperator{\HPR}{HPR}

\newtheorem{theorem}{Theorem}[section]
\newtheorem{definition}[theorem]{Definition}
\newtheorem{lemma}[theorem]{Lemma}
\newtheorem{remark}[theorem]{Remark}
\newtheorem{corollary}[theorem]{Corollary}
\newtheorem{example}[theorem]{Example}
\newtheorem{setup}[theorem]{Setup}

\numberwithin{equation}{section}

\usepackage{bookmark}
\hypersetup{
  pdftitle    = {Risk averse fractional trading using the current drawdown},
  pdfauthor   = {Stanislaus Maier-Paape},
  pdfkeywords = {},
  pdfborder={0 0 0.5}
}

\author{
  \normalsize \textsc{Stanislaus Maier-Paape}\\[-0.2em]
    \small \textit{Institut f\"ur Mathematik, RWTH Aachen,}\\[-0.5em]
    \small \textit{Templergraben 55, D-52062 Aachen, Germany}\\[-0.5em]
    \small \href{mailto:maier@instmath.rwth-aachen.de}{maier@instmath.rwth-aachen.de}
}
\date{
  \vspace{0.25em}
  \normalsize\today
  \vspace{-1cm}
}
\title{
  \vspace{-2cm}
  \Large Risk averse fractional trading using the current drawdown
}

\clearscrheadfoot
\lehead{\pagemark}                   
\rohead{\pagemark}                   
\begin{document}

\maketitle

\vspace*{0.2cm}
\begin{quote}
  \small
  \quoteparagraph{Abstract}
  In this paper the fractional trading ansatz of money management  is reconsidered with special attention to chance and risk parts in the
  goal function of the related optimization problem. By changing the goal function with due regards to other risk measures like
  current drawdowns, the optimal fraction solutions reflect the needs of risk averse investors better than the original optimal $f$
  solution of Vince \cite{vince:pmf90}.

  \quoteparagraph{Keywords} fractional trading, optimal f, current drawdown, terminal wealth relative, risk aversion
\end{quote}



\vspace*{0.5cm}
            \section{Introduction}   \label{sec:introduction}

 The aspects of money and risk management contribute a central scope to investment strategies. Besides the ``modern portfolio theory''
 of Markowitz \cite{markowitz:pfs1991} in particular the methods of fractional trading are well known.

 In the 50's already Kelly \cite{kelly:nii} established a criterion for an asymptotically optimal investment strategy. Kelly as well as  Vince \cite{vince:pmf90} and \cite{vince:mmm92} used the fractional trading ansatz for position sizing of portfolios. In ``fixed fractional trading'' strategies 
 an investor always wants to risk a fixed percentage of his current capital for future investments given some distribution of historic trades of his trading strategy. In Section~\ref{sec:2} we introduce Kelly's and Vince's methods more closely and introduce a common generalization of both models. Both of these methods have in common that their goal function (e.g. $\TWR$=``terminal wealth relative'') solely optimizes wealth growth in the long run, but neglects risk aspects such as the drawdown of the equity curve.

 At this point our research sets in. With one of our results (Theorem~\ref{theo:EU_smallf} and \ref{theo:ED_smallf}) it is possible to split the goal function 
 of Vince into ``chance'' and ``risk'' parts which are easily calculable by an easy representation. In simplified terms, the usual $\TWR$ goal function now 
 takes the form of the expectation of a logarithmic chance --- risk relation
 \begin{align}\label{eq:logE} 
   \E\rk{\log\rk{\frac{\text{chance}}{\text{risk}}}}\,=\,\E\Big(\log\rk{\text{chance}}\Big)\,-\,\E\Big(\,\big|\log\rk{\text{risk}}\big|\,\Big)\,.
 \end{align}
 Moreover, further research (see Section~\ref{sec:5}) revealed an explicitly calculable representation for the expection of new risk measures, namely
 the current drawdown in the framework of fractional trading.

 Having said this, it now seems natural to replace the risk part in \eqref{eq:logE} by the new risk measure of the current drawdown in order to obtain
 a new goal function for fractional trading which fits the needs of risk averse investors much better. This strategy is worked out in Section~\ref{sec:6}
 including existence and uniqueness results for this new risk averse optimal fraction problem.

 The reason such risk averse strategies are deeply needed, lies in the fact that usual optimal $f$ strategies yield not only optimal wealth growth in the
 long run, but also tremendous drawdowns, as shown by empirical simulations in Maier-Paape \cite{maier:optf2015} (see also simulations in section \ref{sec:6}).
 Apparently this problem has also been recognized in the trader community where optimal $f$ strategies are often viewed as ``too risky''
 (cf. van Tharp \cite{tharp:cts01}). The awareness of this problem has also initiated other research to overcome ``too risky'' strategies. For instance,
 Maier-Paape \cite{maier:eto2013}, proved existence and uniqueness of an optimal fraction subject to a risk of ruin constraint. Risk aware strategies in 
 the framework of fractional trading are also discussed in de Prado, Vince and Zhu \cite{prado:orb2013}, and Vince and Zhu \cite{vince:ips2013} suggest 
 to use the inflection point in order to reduce risk. Furthermore a common strategy to overcome tremendous drawdowns is diversification as ascertained by
 Maier-Paape \cite{maier:optf2015} for the Kelly situation.


\vspace*{0.5cm}
      \section{Combing Kelly betting and optimal $\mathbf{f}$ theory} \label{sec:2}    

 In this still introductory section we reconsider two well-known money management strategies, namely the Kelly betting system \cite{ferguson:kbs}, \cite{kelly:nii} and the optimal $f$ model of Vince \cite{vince:pmf90}, \cite{vince:mmm92}. Our intention here is not only to introduce the general concept 
 and notation of fractional trading, but also to find a supermodel which generalizes both of them (which is not obvious). All fractional trading concepts assume that a given trading system offers a series of reproducible profitable trades and ask the question which (fixed) fraction\ $f \in [0,1)$ of the current capital should be invested such that in the long run the wealth growth is optimal with respect to a given goal function. This typically yields an optimization problem in the variable $f$ whose optimal solution is searched for. Both Kelly betting and Vince's optimal $f$ theory are stated in that way.

\vspace*{0.3cm}
\begin{setup}\label{setup:kellybetting} 
  \textbf{(Kelly betting variant)}  \\ [0.1cm]
  Assume a trading system $Y$ with two possible trading result: either one wins $B>0$ with probability $p$ or one loses $-1$ with probability $q=1-p$.
  The trading system should be, profitable, i.e. the expected gain should be positive $\bar{Y}:=p\cdot B-q > 0$.
\end{setup}
 The goal function introduced by Kelly is the so called log--utility function
 \begin{align} \label{eq:log-utility} 
    h(f) := p\cdot\log(1+Bf) + q\cdot\log(1-f) \stackrel{\mbox{\large!}}{=} \max\,, \qquad f \in [0,1)
 \end{align}
 which has to be maximized. The well-known Kelly formula\ $f^{\textnormal{KellyV}}=p-\frac{q}{B}$\ gives the unique solution of \eqref{eq:log-utility}.

\vspace*{0.3cm}
\begin{setup}\label{setup:vince} 
  \textbf{(Vince optimal $\mathbf{f}$ model)} \\ [0.1cm]
  Assume a trading system with absolute trading results\ $t_1,\ldots,t_N \in \R$, $X$ is given with at least one negative trade result.
  Again the trading system should be profitable, i.e. $\bar{X}:=\sum_{i=1}^Nt_i > 0$.
\end{setup}

\vspace*{0.2cm}
 As goal function Vince introduced the so called ``terminal wealth relative''
 \begin{align}\label{eq:TWR-f} 
    \TWR(f) := \prod_{i=1}^N \rk{1+f\frac{t_i}{\hat{t}}} \stackrel{\mbox{\large!}}{=}\max, \qquad f \in [0,1)\,,
 \end{align}
 where\ $\hat{t} = \max\{|t_i|\,:\,t_i<0\} > 0$\ is the maximal loss. The $\TWR$ is the factor between terminal wealth and starting wealth,
 when each of the $N$ trading results occurs exactly once and each time a fraction $f$ of the current capitals is put on risk for the new trade. \vspace*{0.3cm}

 How can one combine these models? The following setup is a generalization of both:
\begin{setup}\label{setup:TWRmodel} 
  \textbf{(general TWR model)}  \\ [0.1cm]
  Assume a trading system $Z$ with absolute trades\ $t_1,\ldots,t_N \in \R$ is given and each trade $t_i$ occurs $N_i\in\N$ times.
  Again we need at least one negative trade and profitability, i.e. $\displaystyle \sum^N_{i=1}\,N_i\,t_i > 0$.
 \end{setup}

 The terminal wealth goal function is easily adapted to $\TWR(f)=\prod_{i=1}^N\rk{1+f\frac{t_i}{\hat{t}}}^{N_i}$ with $\hat{N}:=\sum_{i=1} N_i$.
 Since\ $\TWR(f) > 0$\ for all\ $f \in [0,1)$ the following equivalences are straight forward: \vspace*{-0.3cm}
 \begin{align*}
    \TWR(f)\stackrel{\mbox{\large!}}{=}\max \quad
        &\Leftrightarrow\quad \log\TWR(f)\stackrel{\mbox{\large!}}{=}\max                                                           \\
        &\Leftrightarrow\quad \sum_{i=1}^N N_i\cdot \log\rk{1+f\frac{t_i}{\hat{t}}}\stackrel{\mbox{\large!}}{=}\max                 \\
        &\Leftrightarrow\quad \sum_{i=1}^N \frac{N_i}{\hat{N}}\cdot \log\rk{1+f\frac{t_i}{\hat{t}}}\stackrel{\mbox{\large!}}{=}\max \\
        &\Leftrightarrow\quad \log\Gamma(f)\stackrel{\mbox{\large!}}{=}\max,
 \end{align*}
 where $\Gamma(f)=\prod_{i=1}^N\rk{1+f\frac{t_i}{\hat{t}}}^{p_i}$ is the weighted geometric mean and\ $p_i = \frac{N_i}{\hat{N}}\;\text{for}\ i=1,\ldots,N$
 are the relative frequencies. In this sense trading system $Z$ indeed generalizes both trading systems $Y$ and $X$.

 In particular alternatively to Setup~\ref{setup:TWRmodel} it seems natural to formulate the trading system in a probability setup with trades $t_i$
 which are assumed with a probability $p_i$. This is done in the next section (cf. Setup~\ref{setup:optZ}), where we give an existence and uniqueness results
 for the related optimization problem.


\vspace*{0.5cm}
              \section{Existence of an unique optimal $\mathbf{f}$}\label{sec:3}  

\begin{setup}\label{setup:optZ} 
  Assume a trading system $Z$ with trade results $t_1,\ldots,t_N\in\R\backslash\{0\}$, maximal loss $\hat{t}=\max\{|t_i|\,:\,t_i < 0\}>0$ and relative frequencies
  $p_i=\frac{N_i}{\hat{N}}>0$, where $N_i\in\N$ and $\hat{N}=\sum_{i=1}^NN_i$. Furthermore $Z$ should have positive expectation $\bar{Z}:=\E(Z):=\sum_{i=1}^Np_it_i>0$. \\
\end{setup}
\begin{theorem}\label{theo:optimal_f} 
  Assume Setup~\ref{setup:optZ} holds. Then to optimize the terminal wealth relative
  \begin{align}\label{eq:twr_opt} 
    \TWR(f)=\prod_{i=1}^N\rk{1+f\frac{t_i}{\hat{t}}}^{N_i} \stackrel{\textnormal{\large!}}{=}\max\,, \qquad \text{for}\ f \in [0,1]
  \end{align}
  has a unique solution $f=f^{\textnormal{opt}}\in(0,1)$ which is called \textbf{optimal $\mathbf{f}$}.
\end{theorem}
\begin{proof}
  The proof is along the lines of the ``optimal $f$ lemma'' in \cite{maier:eto2013}:
  \begin{align*}
    \TWR(f)\stackrel{\mbox{\large!}}{=}\max \quad
     &\Longleftrightarrow\quad h(f):=\sum_{i=1}^N p_i\cdot \log\rk{1+f\frac{t_i}{\hat{t}}}\stackrel{\mbox{\large!}}{=}\max \\
     &\Longleftrightarrow\quad 0\stackrel{\mbox{\large!}}{=} h'(f)=\sum_{i=1}^N\frac{p_i\frac{t_i}{\hat{t}}}{1+f\frac{t_i}{\hat{t}}}=\sum_{i=1}^N\frac{p_i}{b_i+f}=:g(f),
  \end{align*}
  where $a_i:=\frac{t_i}{\hat{t}}\in[-1,\infty)\backslash\{0\}$ and $b_i:=\frac{1}{a_i}\in(-\infty,-1]\cup(0,\infty)$. Assume w.l.o.g that $b_i$ are
  ordered and $b_{i_0} = -1$.Then $b_{i_0+1} > 0$. Since $\frac{p_i}{b_i+f}$ is strictly monotone decreasing for $f\neq -b_i$, so is $g(f)$ for $f\neq\{-b_i\,:\,i=1,\ldots,N\}$. This yields existence
  of exactly one zero $f^{\ast}$ of $g$ in $(-b_{i_0+1},1)$ and since $g(0)=h'(0)=\frac{1}{\hat{t}}\sum_{i=1}^Np_it_i>0$ we have $f^{\ast}>0$.
  Hence $f^{\textnormal{opt}}=f^{\ast} \in (0,1)$ is the unique solution of \eqref{eq:twr_opt} (see Figure \ref{figfOpt}). \\
\end{proof} \vspace*{-0.5cm}



  \begin{figure}[htb]
    \centering
    \begin{minipage}[c]{0.8\linewidth}
        \centering
        \includegraphics[width=0.7\textwidth]{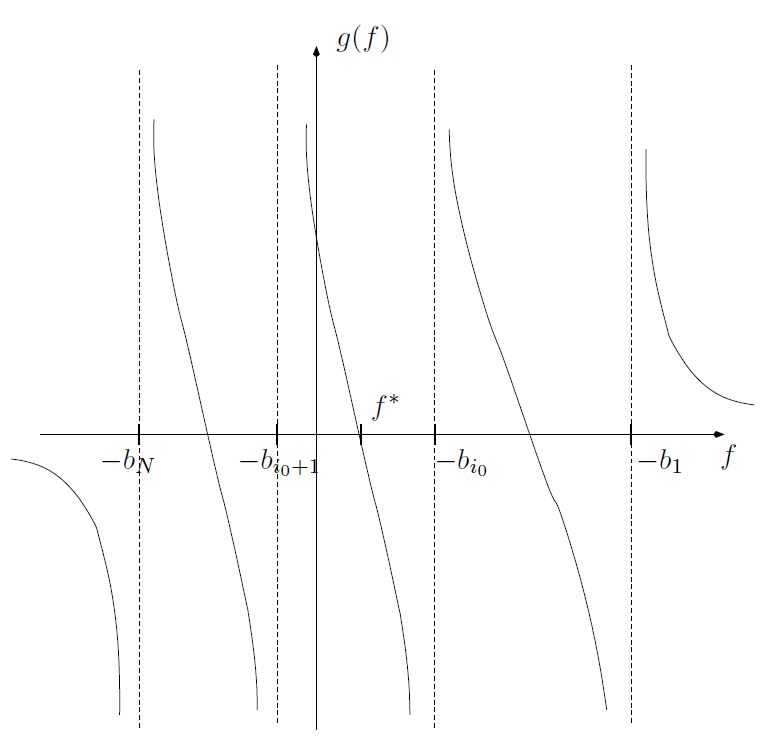}

    \end{minipage}
    \caption{Zeros of $g$ yielding the existence of  $f^{\textnormal{opt}}$}
    \label{figfOpt}
  \end{figure}

\vspace*{0.5cm}

\begin{remark}\label{rem:optimal_f} 
  Theorem~\ref{theo:optimal_f} holds true even if $p_i>0$ are probabilities, $\sum_{i=1}^Np_i=1$ and
  \begin{align}\label{eq:Gamma f} 
    \Gamma(f)=\prod_{i=1}^N\rk{1+f\frac{t_i}{\hat{t}}}^{p_i}\stackrel{\textnormal{\large!}}{=}\max\,, \qquad \text{for}\ f \in [0,1]\,.
  \end{align}

   I.e. the optimization problem \eqref{eq:Gamma f} has an optimal  $f$ solution as well.
   It is important to note that the result so far uses no probability theory at all.
\end{remark}

\newpage


\vspace*{0.5cm}
          \section{Randomly drawing trades}\label{sec:4}   

\begin{setup}\label{setup:Z} 
  Assume a trading system with trade results $t_1,\ldots,t_N\in\R\backslash\{0\}$ and with maximal loss $\hat{t}=\max\{|t_i|\,:\,t_i<0\}>0$.
  Each trade $t_i$ has a probability of $p_i>0$, with $\sum_{i=1}^Np_i=1$. Drawing randomly and independent $M\in\N$ times from this distribution
  results in a  probability space $\Omega^{(M)}:=\{\omega=(\omega_1,\ldots,\omega_M)\,:\,\omega_i\in\{1,\ldots,N\}\}$ and a terminal wealth relative
  (for fractional trading with fraction $f$)
  \begin{align}\label{eq:twr_omega} 
    \TWR_1^M(f,\omega):=\prod_{j=1}^M\rk{1+f\frac{t_{\omega_j}}{\hat{t}}},\quad f\in[0,1).
  \end{align}
\end{setup}
\begin{theorem}\label{theo:EZ} 
  The random variable $\cZ^{(M)}(f,\omega):=\log(\TWR_1^M(f,\omega))$ has expectation value
  \begin{align}\label{eq:EZ} 
    \E(\cZ^{(M)}(f,\cdot))=M\cdot\log\Gamma(f),\quad\text{ for all $f\in[0,1)$,}
  \end{align}
  where $\Gamma(f)=\prod_{i=1}^N\rk{1+f\frac{t_i}{\hat{t}}}^{p_i}$ is the weighted geometric mean of the holding period returns\
  $\HPR_i:=1+f\frac{t_i}{\hat{t}} > 0$\ for all\ $f \in [0,1)$.
\end{theorem}
\begin{proof}
  Case $M=1$: Here $\cZ^{(1)}(f,\omega_1)=\log\rk{1+f\frac{t_{\omega_1}}{\hat{t}}}$ and
  \begin{align*}
    \E(\cZ^{(1)}(f,\cdot))
     &=\sum_{i=1}^Np_i\log\rk{1+f\frac{t_i}{\hat{t}}}
      =\log\ek{\prod_{i=1}^N\rk{1+f\frac{t_i}{\hat{t}}}^{p_i}}
      =\log\Gamma(f).
  \end{align*}
  The induction step $M-1\to M$: Using $\P(\{\omega\})=\P^{(M)}(\{\omega\})=\prod_{i=1}^Mp_{\omega_i}$ for $\omega \in \Omega^{(M)}$ and
  $\omega^{(M-1)}:=(\omega_1,\ldots,\omega_{M-1})$ we get
  \begin{align*}
    & \E(\cZ^{(M)}(f,\cdot))\,=\,\sum_{\omega\in\Omega^{(M)}}\P(\{\omega\})\log\rk{\prod_{j=1}^M\rk{1+f\frac{t_{\omega_j}}{\hat{t}}}}                      \\
    & =\,\sum_{\omega^{(M-1)}\in\Omega^{(M-1)}}\; \sum_{\omega_M=1}^N\P^{(M-1)}(\{\omega^{(M-1)}\})\cdot p_{\omega_M}\cdot\ek{\log\rk{\prod_{j=1}^{M-1}
                                                                       \rk{1+f\frac{t_{\omega_j}}{\hat{t}}}}+\log\rk{1+f\frac{t_{\omega_M}}{\hat{t}}}}   \\
    & =\,\sum_{\omega^{(M-1)}\in\Omega^{(M-1)}}\P^{(M-1)}(\{\omega^{(M-1)}\})\cdot\log\TWR_1^{M-1}(f,\omega^{(M-1)})+\sum_{\omega_M=1}^Np_{\omega_M}
                                                                       \log\rk{1+f\frac{t_{\omega_M}}{\hat{t}}}                                          \\ \\
    & \stackrel{\text{Case 1}}{=}\E(\cZ^{(M-1)}(f,\cdot))+\log\Gamma(f)=M\cdot\log\Gamma(f)
  \end{align*}
  by induction.
\end{proof}

\vspace*{0.2cm}
 As a next step, we want to split up the random variable $\cZ^{(M)}(f,\cdot)$ into \textbf{chance} and \textbf{risk} part. \\ [0.2cm]
 Since\ $\TWR^M_1(f,\omega)>1$\ corresponds to a winning trade series\ $t_{\omega_1},\ldots,t_{\omega_M}$\ and \\
 $\TWR^M_1(f,\omega)<1$ analogously corresponds to a loosing trade series we define the random variables:
\begin{definition} 
  \textbf{Up-trade log series:}
  \begin{align}\label{eq:uptrade log} 
    \cU^{(M)}(f,\omega):=\log(\max\{1,\TWR_1^M(f,\omega)\}) \ge 0.
  \end{align}
  \textbf{Down-trade log series:}
  \begin{align}\label{eq:downtrade log} 
    \cD^{(M)}(f,\omega):=\log(\min\{1,\TWR_1^M(f,\omega)\}) \le 0.
  \end{align}
\end{definition}
\vspace*{0.2cm}
 Clearly $\cU^{(M)}(f,\omega)+\cD^{(M)}(f,\omega)=\cZ^{(M)}(f,\omega)$. Hence by Theorem~\ref{theo:EZ} we get
\begin{corollary}\label{cor:EU+ED} 
  For $f\in[0,1)$
  \begin{align}\label{eq:sum_down_uptrade} 
     \E(\cU^{(M)}(f,\cdot))+\E(\cD^{(M)}(f,\cdot)) = M\log\Gamma(f)
  \end{align}
  holds.
\end{corollary}
 The rest of this section is devoted to find explicitly calculable formulas for $\E(\cU^{(M)}(f,\cdot))$ and $\E(\cD^{(M)}(f,\cdot))$. By definition
 \begin{align}\label{eq:EU} 
   \E(\cU^{(M)}(f,\cdot)) = \sum_{\omega:\TWR_1^M(f,\omega)>1}\P(\{\omega\})\cdot\log(\TWR_1^M(f,\omega)).
 \end{align}
 Assume $\omega=(\omega_1,\ldots,\omega_M)\in\Omega^{(M)}:=\{1,\ldots,N\}^M$ is fixed for the moment and the random variable $X_1$ counts how many of the $\omega_j$ are equal
 to $1$. I.e. $X_1(\omega)=x_1$ if $x_1$ of the $\omega_j$'s in $\omega$ are equal to $1$. With similar counting random variables $X_2,\ldots,X_N$ we obtain counts $x_i \ge 0$
 and thus
 \begin{align}\label{eq:X_i} 
    X_1(\omega) = x_1,\ X_2(\omega)=x_2,\ \ldots,\ X_N(\omega)=x_N
 \end{align}
 with obviously $\sum_{i=1}^Nx_i=M$. Hence for this fixed $\omega$ we get
 \begin{align}\label{eq:TWR_1} 
    \TWR_1^M(f,\omega)=\prod_{j=1}^M\rk{1+f\frac{t_{\omega_j}}{\hat{t}}}=\prod_{i=1}^N\rk{1+f\frac{t_{i}}{\hat{t}}}^{x_i}.
 \end{align}
 Therefore the condition on $\omega$ in the sum \eqref{eq:EU} is equivalent to
 \begin{align}\label{eq:TWR_1_equiv} 
    \TWR_1^M(f,\omega)>1
      \;\Longleftrightarrow\; \log \TWR_1^M(f,\omega)>0
      \;\Longleftrightarrow\; \sum_{i=1}^Nx_i\log\rk{1+f\frac{t_i}{\hat{t}}}>0.
 \end{align}
 To better understand the last sum, we use Taylor expansion to obtain
\begin{lemma}\label{lem:small_f} 
  Let real numbers $\hat{t}>0$, $x_i\geq 0$ with $\sum_{i=1}^Nx_i=M>0$ and $t_i\neq 0$ for $i=1,\ldots,N$ be given. Then the following holds:
  \begin{enumerate}
    \item $\sum_{i=1}^Nx_it_i>0\;\Longleftrightarrow\; h(f):=\sum_{i=1}^Nx_i\log\rk{1+f\frac{t_i}{\hat{t}}}>0$ for all sufficiently small    \\ [0.05cm] 
          $f>0$,                                                                                                                             \\ [-0.15cm]
    \item $\sum_{i=1}^Nx_it_i\leq 0\;\Longleftrightarrow\; h(f)=\sum_{i=1}^Nx_i\log\rk{1+f\frac{t_i}{\hat{t}}}<0$ for all sufficiently small \\ [0.05cm] 
          $f>0$.
  \end{enumerate}
\end{lemma}

\begin{proof}
  {\glqq}$\Rightarrow${\grqq}:
  An easy calculation shows $h'(0)=\frac{1}{\hat{t}}\sum_{i=1}^Nx_it_i$ and $h''(0)=\frac{-1}{\hat{t}^2}\sum_{i=1}^Nx_it_i^2<0$ yielding this
  direction for (a) and (b) since $h(0) = 0$.
\\ [0.2cm]
  {\glqq}$\Leftarrow${\grqq}:
  From the above we conclude that no matter what\ $\sum_{i=1}^Nx_it_i$\ is, always\\ $\sum_{i=1}^Nx_i \log\rk{1+f\frac{t_i}{\hat{t}}}\neq 0$
  for $f>0$ sufficiently small holds. The claim of the backward direction now follows by contradiction.
\end{proof}

 Using Lemma~\ref{lem:small_f} we hence can restate \eqref{eq:TWR_1_equiv}
\begin{align}\label{eq:TWR_1_equiv_smallf} 
  \TWR_1^M(f,\omega)>1\text{for $f > 0$ sufficiently small} \quad\Longleftrightarrow\quad \sum_{i=1}^Nx_it_i > 0.
\end{align}
 After all these preliminaries, we may now state the first main result.
\begin{theorem}\label{theo:EU_smallf} 
  Let a trading system as in Setup~\ref{setup:Z} with fixed\ $N,M \in \N$\ be given. \\
  Then for all sufficiently small $f>0$ the following holds:
  \begin{align}\label{eq:EU_smallf} 
     \E\big(\cU^{(M)}(f,\cdot)\big)\,=\,u^{(M)}(f)\,:=\,\sum_{n=1}^N U_n^{(M,N)}\cdot\log\rk{1+f\frac{t_n}{\hat{t}}},
  \end{align}
  where
  \begin{align}\label{eq:UnMN} 
    U_n^{(M,N)} := \sum_{\substack{(x_1,\ldots,x_N)\in\N_0^N  \\ \sum\limits_{i=1}^Nx_i=M,\ \sum\limits_{i=1}^Nx_it_i>0}}
       p_1^{x_1} \cdots p_N^{x_N} \cdot \binom{M}{x_1\; x_2\cdots x_N}\cdot x_n \geq 0
  \end{align}
  and\ $\displaystyle \binom{M}{x_1\; x_2\cdots x_N} = \frac{M!}{x_1!x_2!\cdots x_N!}$\ is the multinomial coefficient.
\end{theorem}
\begin{proof}
  Starting with \eqref{eq:EU} and using \eqref{eq:X_i} and \eqref{eq:TWR_1_equiv_smallf} we get for sufficiently small $f > 0$
  \begin{align*}
    \E(\cU^{(M)}(f,\cdot))\,=\,
       \sum_{\substack{(x_1,\ldots,x_N)\in\N_0^N \\ \sum\limits_{i=1}^Nx_i\,=\,M}} \hspace{2mm} \sum_{\substack{\omega:X_1(\omega)\,=\,x_1,\ldots,X_N(\omega)=x_N \\
       \sum\limits_{i=1}^Nx_it_i>0}}
    \P(\{\omega\})\cdot \log\rk{\TWR_1^M(f,\omega)}.
  \end{align*}
  Since there are\ $\binom{M}{x_1\; x_2\cdots x_N}=\frac{M!}{x_1!x_2!\cdots x_N!}$\ many\ $\omega\in\Omega^{(M)}$\ for which\
  $X_1(\omega)=x_1,\ldots,X_N(\omega)=x_N$\ holds we furthermore get from \eqref{eq:TWR_1}
  \begin{align*}
    \E(\cU^{(M)}(f,\cdot))
     &=\,\sum_{\substack{(x_1,\ldots,x_N)\in\N_0^N \\ \sum\limits_{i=1}^Nx_i\,=\,M,\,\sum\limits_{i=1}^Nx_it_i>0}}
       p_1^{x_1}\cdots p_N^{x_N} \cdot \binom{M}{x_1\; x_2\cdots x_N} \sum_{n=1}^Nx_n \cdot \log\rk{1+f\frac{t_n}{\hat{t}}} \\ \\
     &=\,\sum_{n=1}^N U_n^{(M,N)} \cdot \log\rk{1+f\frac{t_n}{\hat{t}}}
  \end{align*}
  as claimed.
\end{proof}
 A similar result holds for $\E(\cD^{(M)}(f,\cdot))$.
\begin{theorem}\label{theo:ED_smallf} 
  In the situation of Theorem~\ref{theo:EU_smallf} for sufficiently small $f>0$
  \begin{align}\label{eq:ED_smallf} 
     \E(\cD^{(M)}(f,\cdot))\,=\,d^{(M)}(f)\,:=\,\sum_{n=1}^N D_n^{(M,N)} \cdot \log\rk{1+f\frac{t_n}{\hat{t}}},
  \end{align}
  holds, where
  \begin{align}\label{eq:DnMN} 
    D_n^{(M,N)}
      := \sum_{\substack{(x_1,\ldots,x_N)\in\N_0^N \\ \sum\limits_{i=1}^Nx_i\,=\,M,\,\sum\limits_{i=1}^Nx_it_i\leq 0}}
           p_1^{x_1}\cdots p_N^{x_N}\cdot \binom{M}{x_1\; x_2\cdots x_N}\cdot x_n \geq 0.
  \end{align}
\end{theorem}
\begin{proof}
  By definition
  \begin{align*}
     \E(\cD^{(M)}(f,\cdot)) = \sum_{\omega:\TWR_1^M(f,\omega)<1}\P(\{\omega\})\cdot \log\rk{\TWR_1^M(f,\omega)}.
  \end{align*}

  The arguments given in the proof of Theorem~\ref{theo:EU_smallf} apply similarly, where instead of \eqref{eq:TWR_1_equiv_smallf} we use
  Lemma~\ref{lem:small_f} (b) to get
  \begin{align}\label{eq:TWR_1_equiv_smallf_b} 
    \TWR_1^M(f,\omega) < 1 \quad\text{for all $f>0$ sufficiently small} \;\Longleftrightarrow\;  \sum_{i=1}^Nx_it_i\,\leq\,0\,.
  \end{align}
\end{proof}

\vspace*{0.2cm}
\begin{remark} 
  Using the well-known fact from multinomial distributions
  \begin{align*}
    \sum_{\substack{(x_1,\ldots,x_N)\in\N_0^N \\ \sum\limits_{i=1}^Nx_i\,=\,M}} \hspace{2mm}
        p_1^{x_1}\cdots p_N^{x_N}\cdot \binom{M}{x_1\; x_2\cdots x_N} = (p_1+\ldots+p_N)^M = 1
  \end{align*}
  it immediately follows that
  \begin{align*}
    \sum_{\substack{(x_1,\ldots,x_N)\in\N_0^N \\ \sum\limits_{i=1}^Nx_i\,=\,M}} \hspace{2mm}
        p_1^{x_1}\cdots p_N^{x_N}\cdot \binom{M}{x_1\; x_2\cdots x_N}x_n = p_n\cdot M \quad \text{for all $n = 1,\ldots,N$}
  \end{align*}
  yielding (again) with Theorem~\ref{theo:EU_smallf} and \ref{theo:ED_smallf}
  \begin{align*}
    \E(\cU^{(M)}(f,\cdot))+\E(\cD^{(M)}(f,\cdot))
      = \sum_{n=1}^Np_n\cdot M\cdot \log\rk{1+f\frac{t_n}{\hat{t}}}
      = M\cdot \log\Gamma(f).
  \end{align*}
\end{remark}
 At next we want to apply our theory to the $2:1$ toss game, where a coin is thrown. In case coin shows head, the stake is doubled,
 whereas in case of tail it is lost. \vspace*{0.2cm}

\begin{example}\label{ex:toss_game_1} 
  ($2:1$ toss game; $M=3$)\\
  Here $N=2$, $p_i=\frac{1}{2}$, $t_1=-1$, $t_2=2$ and $\hat{t}=1$.
  In this case \eqref{eq:UnMN} simplifies to
  \begin{align*}
    U_1^{(M,2)}=\frac{1}{2^M}\sum_{\substack{k=0\\k\cdot t_1+(M-k)t_2>0}}^M\binom{M}{k}\cdot k
    \quad\text{and}\quad
    U_2^{(M,2)}=\frac{1}{2^M}\sum_{\substack{k=0\\k\cdot t_1+(M-k)t_2>0}}^M\binom{M}{k}\cdot (M-k).
  \end{align*}
  Hence with \eqref{eq:EU_smallf} for $f>0$ sufficiently small
  \begin{align*}
    \E(\cU^{(M)}(f,\cdot))
      = \frac{1}{2^M}\sum_{\substack{k=0\\k\cdot t_1+(M-k)t_2>0}}^M\binom{M}{k}
        \ek{k\log(1+ft_1) + (M-k)\log(1+ft_2)}
  \end{align*}
  and analogously
  \begin{align*}
    \E(\cD^{(M)}(f,\cdot))
      = \frac{1}{2^M}\sum_{\substack{k=0\\k\cdot t_1+(M-k)t_2\leq 0}}^M\binom{M}{k}
        \ek{k\log(1+ft_1) + (M-k)\log(1+ft_2)}
  \end{align*}
  Letting now $M=3$ from $t_2=2$ and $t_1=-1$ we get $kt_1+(M-k)t_2>0$ for $k=0$ and $k=1$ only.
  Therefore
  \begin{align} 
    \notag
     \E(\cU^{(3)}(f,\cdot))
       &= \frac{1}{2^3}\ek{3\log(1+2f)+\binom{3}{1}(\log(1-f)+2\log(1+2f))}\\
     \label{eq:EU_smallf cdot}
       &= \frac{1}{2^3}\ek{3\log(1-f)+9\log(1+2f)}\\
     \notag
  \end{align}
  and similarly
  \begin{align}\label{eq:ED_smallf cdot} 
    \E(\cD^{(3)}(f,\cdot)) &= \frac{1}{2^3}\ek{9\log(1-f)+3\log(1+2f)}.
  \end{align}
  for\ $f>0$\ sufficiently small. In Figure \ref{figEcur} one can see that these approximations are quite accurate
  up to\ $f=0.85$.
\end{example}



  \begin{figure}[htb]
    \centering
    \begin{minipage}[c]{1\linewidth}
        \centering
       \includegraphics[width=1.0\textwidth]{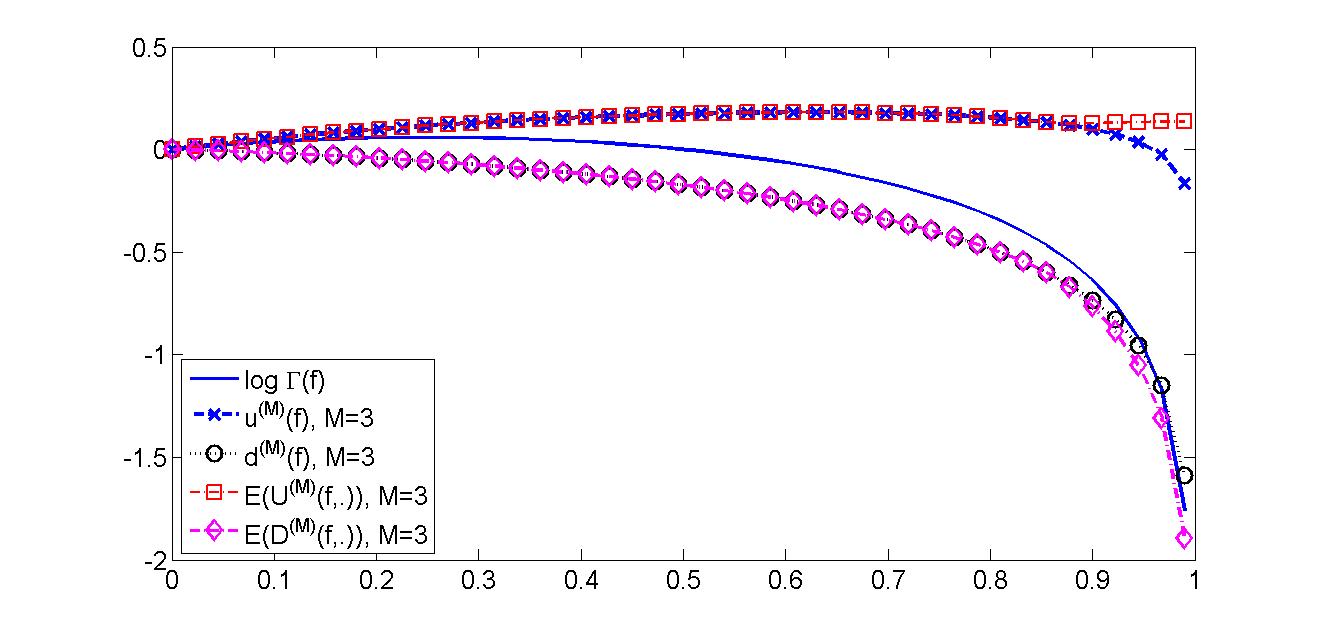}
    \end{minipage}
    \caption{$\E(\cU^{(3)}(f,\cdot))$ and $\E(\cD^{(3)}(f,\cdot))$ with their approximations of
             \eqref{eq:EU_smallf}  and  \eqref{eq:ED_smallf} }
    \label{figEcur}
  \end{figure}




\newpage
          \section{The current drawdown}\label{sec:5}   

 We keep discussing the trading system with trades $t_1,\ldots,t_N\in\R\backslash\{0\}$ and probabilities $p_1,\ldots,p_N$ from Setup~\ref{setup:Z} and draw
 randomly and independent $M\in\N$ times from that distribution. At next we want to investigate the resulting terminal wealth relative from fractional trading
\begin{align*}
  \TWR_1^M(f,\omega) = \prod_{j=1}^M\rk{1+f\frac{t_{\omega_j}}{\hat{t}}}, \quad f\in[0,1),\ \omega\in\Omega^{(M)}=\{1,\ldots,N\}^M
\end{align*}
 from \eqref{eq:twr_omega} with respect to the \textit{current drawdown} realized after the $M$th draw. \\
 More generally, in the following we will use
 \begin{align*}
    \TWR^n_m(f,\omega)\,:=\,\prod\limits^n_{j=m}\,\rk{1+f\frac{t_{\omega_j}}{\hat{t}}}\,.
 \end{align*}

 The idea here is that $\TWR^n_1(f,\omega)$ is viewed as a discrete {\glqq}equity curve{\grqq} for time $n$ (with $f$ and $\omega$ fixed).
 The current drawdown $\log$-series is the logarithm of the drawdown of this equity curve realized from the maximum of the curve til the end (time $M$).
 As we will see below, this series is the counter part of the \textit{runup} (cf. Figure \ref{figlog-series}).


  \begin{figure}[htb]
    \centering
    \begin{minipage}[c]{0.5\linewidth}
        \centering
        \includegraphics[width=0.9\textwidth]{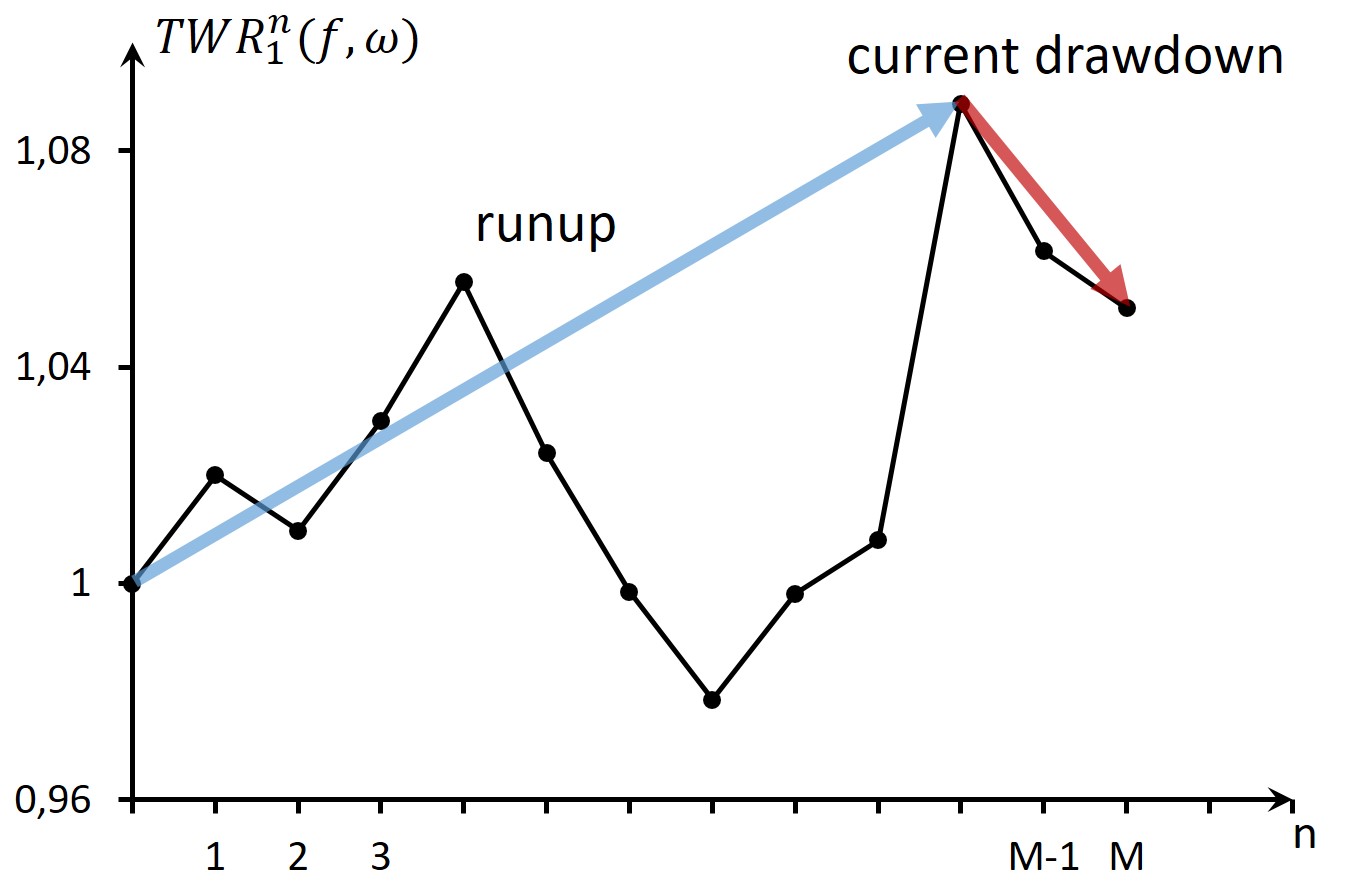}
    \end{minipage}
\hspace{-0.6cm}
    \begin{minipage}[c]{0.5\linewidth}
        \centering
        \includegraphics[width=0.9\textwidth]{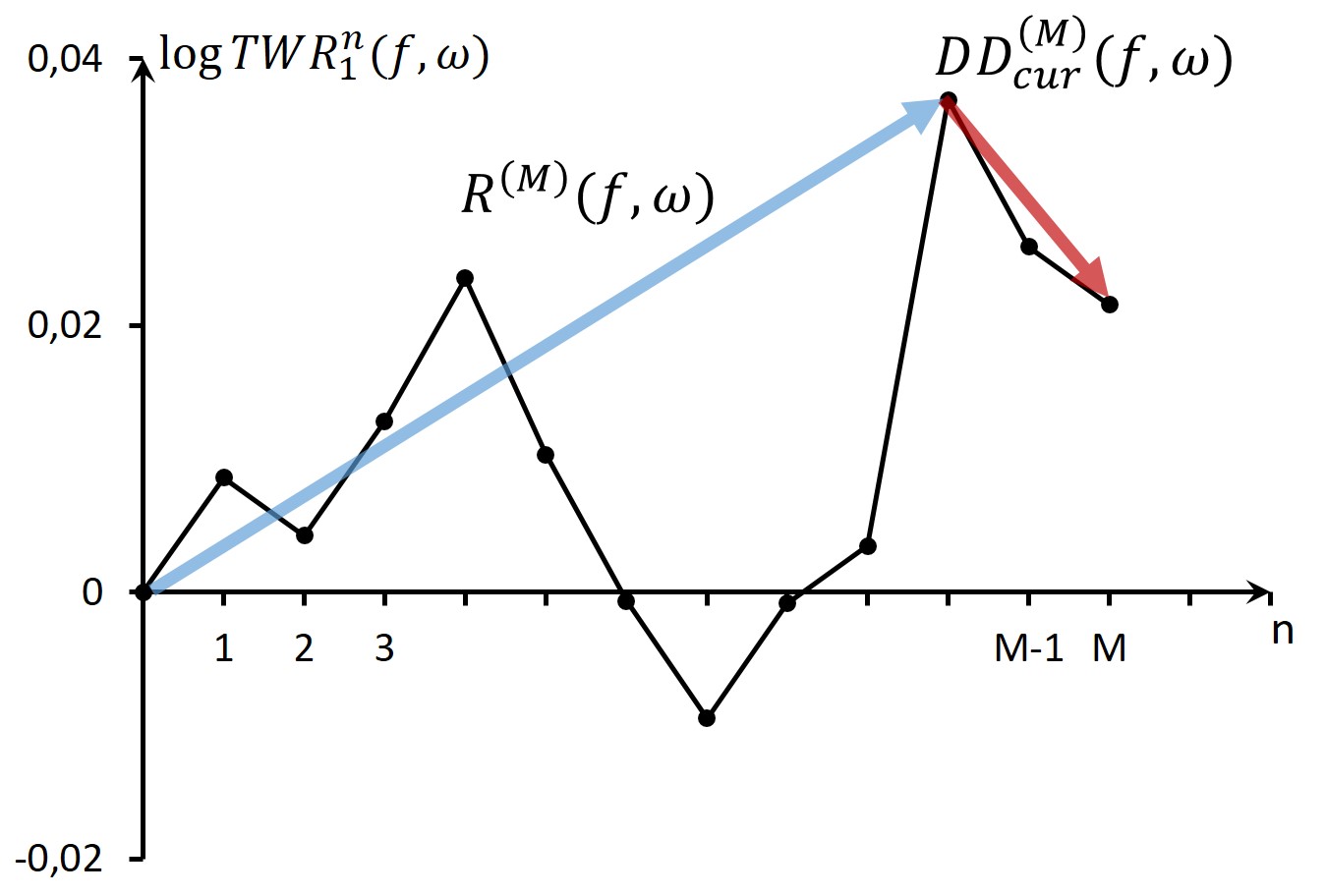}
    \end{minipage}
    \caption{In the left figure the run-up and the current drawdown is plotted for an instance
             of the \text{TWR} ``equity''--curve and to the right are their $\log$ series. }
  \label{figlog-series}
  \end{figure}



\begin{definition}\label{def:logseries} 
  The \textbf{current drawdown log series} is set to
  \begin{align*}
    \cDD^{(M)}(f,\omega):=\log\rk{\min_{1\,\leq\,\ell\,\leq\,M}\min\{1,\TWR_\ell^M(f,\omega)\}} \le 0\,,
  \end{align*}
  and the \textbf{run-up log series} is defined as
  \begin{align*}
    \cR^{(M)}(f,\omega):=\log\rk{\max_{1\,\leq\,\ell\,\leq\,M}\max\{1,\TWR_1^\ell(f,\omega)\}} \ge 0\,.
  \end{align*}
\end{definition}
 The corresponding trade series are connected in that way that the current drawdown starts after the run-up has stopped.
 To make that more precise, we fix that $\ell$ where the run-up topped.
\begin{definition}\label{def:l*} 
  For fixed $\omega\in\Omega^{(M)}$, $f\in[0,1)$ define $\ell^{\ast}\!=\ell^{\ast}(f,\omega)\in\{0,\ldots,M\}$ with
  \begin{enumerate}
    \item $\ell^\ast=0$ in case\ $\max\limits_{1\,\leq\,\ell\,\leq\,M}\TWR_1^{\ell}(f,\omega)\leq 1$
    \item and otherwise choose $\ell^\ast\!\in\{1,\ldots,M\}$ such that
          \begin{align}\label{eq:TWR_1_l*} 
             \TWR_1^{\ell^\ast}(f,\omega)=\max_{1\leq\ell\leq M}\TWR_1^\ell(f,\omega)>1,
          \end{align}
      where $\ell^\ast$ should be minimal with that property.
  \end{enumerate}
\end{definition}
 By definition one easily sees
 \begin{align}\label{eq:DD}
      \cDD^{(M)}(f,\omega) & = \begin{cases}
                                  \log\TWR_{\ell^\ast+1}^M(f,\omega), &\text{in case $\ell^\ast\!< M$,} \\
                                     0,                               &\text{in case $\ell^\ast\!= M$},
                               \end{cases} \\
   \intertext{and}
              \label{eq:R}
      \cR^{(M)}(f,\omega)  & = \begin{cases}
                                  \log\TWR_1^{\ell^\ast}(f,\omega), &\text{in case $\ell^\ast\!\geq 1$,} \\
                                     0,                             &\text{in case $\ell^\ast\!= 0$}.
                               \end{cases}
 \end{align} 
 As in Section~\ref{sec:4} we immediately get $\cDD^{(M)}(f,\omega)+\cR^{(M)}(f,\omega)=\cZ^{(M)}(f,\omega)$ and therefore by Theorem~\ref{theo:EZ}:
\vspace*{0.2cm}
\begin{corollary} 
  For $f\in[0,1)$
  \begin{align}\label{eq:EDD+ER} 
    \E(\cDD^{(M)}(f,\cdot)) + \E(\cR^{(M)}(f,\cdot)) = M\log\Gamma(f)
  \end{align}
  holds.
\end{corollary}
 Again explicit formulas for the expectation of $\cDD^{(M)}$ and $\cR^{(M)}$ are of interest. \\
 By definition and with \eqref{eq:DD}
 \begin{align}\label{eq:EDD_def} 
  \E(\cDD^{(M)}(f,\cdot))
   = \sum_{\ell=0}^{M-1}\sum_{\substack{\omega\in\Omega^{(M)}\\\ell^\ast(f,\omega)=\ell}}
     \P(\{\omega\})\cdot\log\TWR_{\ell+1}^M(f,\omega)
 \end{align}
 Before we proceed with this calculation we need to discuss $\ell^\ast\!=\ell^\ast(f,\omega)$ further for some fixed $\omega$.
 By Definition~\ref{def:l*}, in case $\ell^{\ast}\!\ge1$, we have
 \begin{align}
        \TWR_k^{\ell^\ast}(f,\omega)             & > 1\quad\text{for $k=1,\ldots,\ell^\ast$}                 \\
   \intertext{since $\ell^\ast$ is the first time the run-up topped and, in case $\ell^{\ast}\!<M$}
        \TWR_{\ell^\ast+1}^{\tilde{k}}(f,\omega) & \leq 1\quad\text{for $\tilde{k}=\ell^\ast+1,\ldots,M$}.
 \end{align} 
 For instance the last inequality may for all sufficiently small $f>0$ be rephrased as
\begin{align} 
  \notag
        \TWR_{\ell^\ast+1}^{\tilde{k}}(f,\omega)\leq 1
          \quad &\Longleftrightarrow \quad \log\TWR_{\ell^\ast+1}^{\tilde{k}}(f,\omega)\leq 0                              \\
  \notag        &\Longleftrightarrow \quad \sum_{j=\ell^\ast+1}^{\tilde{k}}\log\rk{1+f\frac{t_{\omega_j}}{\hat{t}}}\leq 0  \\
  \label{eq:5.8}
                &\Longleftrightarrow \quad \sum_{j=\ell^\ast+1}^{\tilde{k}}t_{\omega_j}\leq 0
\end{align}
 by an argument similar to Lemma~\ref{lem:small_f}. Analogously one finds
\begin{align}\label{eq:twr k} 
   \TWR_k^{\ell^\ast}(f,\omega)> 1\quad\text{for all $f>0$ sufficiently small} \quad &\Longleftrightarrow \quad \sum_{j=k}^{\ell^\ast}t_{\omega_j} > 0.
\end{align}

 We may now state the main result on the expectation of the current drawdown.
\begin{theorem}\label{theo:EDD} 
  Let a trading system as in Setup~\ref{setup:Z} with fixed $N,M \in \N$ be given. \\
  Then for all sufficiently small $f>0$  the following holds:
  \begin{align}\label{eq:EDD} 
    \E(\cDD^{(M)}(f,\cdot)) = d_{\text{cur}}^{(M)}(f) := \sum_{n=1}^N\;\rk{\sum_{\ell=0}^M\Lambda_n^{(\ell,M,N)}}\cdot \log\rk{1+f\frac{t_n}{\hat{t}}}
  \end{align}
  where $\Lambda_n^{(M,M,N)}:=0$ and for $\ell\in\{0,1,\ldots,M-1\}$ the constants $\Lambda^{(\ell,M,N)}_n \ge 0$ are defined by
  \begin{align}\label{eq:Lambda} 
    \Lambda_n^{(\ell,M,N)} :=\hspace{-0.3cm}\sum_{\substack{\omega\in\Omega^{(M)}                                                                          \\
                                                            \sum\limits_{j=k}^\ell t_{\omega_j}>0\text{ for $k=1,\ldots,\ell$}                             \\ \\
                                                            \sum\limits_{j=\ell+1}^{\tilde{k}} t_{\omega_j}\leq 0\text{ for $\tilde{k}=\ell+1,\ldots,M$}}} \hspace{-0.5cm}
    \P(\{\omega\})\cdot \#\{\omega_i=n\,|\,i\geq \ell+1\}.
  \end{align}
\end{theorem}
\vspace*{0.2cm}
\begin{proof}
  Starting with \eqref{eq:EDD_def} we get
  \begin{align*}
    \E(\cDD^{(M)}(f,\cdot)) = \sum_{\ell=0}^{M-1} \sum_{\substack{\omega\in\Omega^{(M)} \\ \ell^\ast(f,\omega)=\ell}}
     \P(\{\omega\})\cdot\sum_{i=\ell+1}^M\log\rk{1+f\frac{t_{\omega_i}}{\hat{t}}}
  \end{align*}
  and by \eqref{eq:5.8} and \eqref{eq:twr k} for all $f>0$ sufficiently small
  \begin{align*}
    &\E(\cDD^{(M)}(f,\cdot))\,=\,\sum_{\ell=0}^{M-1}
         \sum_{\substack{\omega\in\Omega^{(M)}\\ \sum\limits_{j=k}^\ell t_{\omega_j}>0\text{ for $k=1,\ldots,\ell$}\\ \sum\limits_{j=\ell+1}^{\tilde{k}} t_{\omega_j}\leq
                         0\text{ for $\tilde{k}=\ell+1,\ldots,M$}}}\P(\{\omega\})\cdot\sum_{i=\ell+1}^M\log\rk{1+f\frac{t_{\omega_i}}{\hat{t}}}                            \\
    &=\,\sum_{\ell=0}^{M-1}
        \sum_{\substack{\omega\in\Omega^{(M)}\\ \sum\limits_{j=k}^\ell t_{\omega_j}>0\text{ for $k=1,\ldots,\ell$}\\ \sum\limits_{j=\ell+1}^{\tilde{k}} t_{\omega_j}\leq
                         0\text{ for $\tilde{k}=\ell+1,\ldots,M$}}}\P(\{\omega\})\cdot \sum_{n=1}^N\#\{\omega_i=n\,|\,i\geq\ell+1\}\log\rk{1+f\frac{t_n}{\hat{t}}}         \\ \\
    &=\,\sum_{n=1}^N\,\sum_{\ell=0}^{M-1}\Lambda_n^{(\ell,M,N)}\cdot \log\rk{1+f\frac{t_n}{\hat{t}}}
  \end{align*}

  and \eqref{eq:EDD} follows since $\Lambda^{(M,M,N)}_n = 0$.
\end{proof}
\vspace*{0.3cm}
 The  same reasoning yields:
\begin{theorem}\label{theo:ER} 
  In the situation of Theorem~\ref{theo:EDD} for all sufficiently small $f>0$
  \begin{align}\label{eq:ER} 
    \E(\cR^{(M)}(f,\cdot)) = r^{(M)}(f) := \sum_{n=1}^N\,\rk{\sum_{\ell=0}^MR_n^{(\ell,M,N)}}\cdot \log\rk{1+f\frac{t_n}{\hat{t}}}
  \end{align}
  holds, where $R_n^{(0,M,N)}:=0$ and for $\ell\in\{1,\ldots,M\}$ the constants $R_n^{(\ell,M,N)}\ge0$ are given as
  \begin{align}\label{eq:R_n^lMN} 
    R_n^{(\ell,M,N)} := \hspace{-0.3cm}\sum_{\substack{\omega\in\Omega^{(M)}\\ \sum\limits_{j=k}^\ell t_{\omega_j}\,>\,0\,\text{for $k=1,\ldots,\ell$} \\
                                                                               \sum\limits_{j=\ell+1}^{\tilde{k}} t_{\omega_j}\,\leq\,0\,\text{for $\tilde{k}=\ell+1,\ldots,M$}}}\!\!\!
    \P(\{\omega\})\cdot \#\{\omega_i=n\,|\,i\leq\ell\}.
  \end{align}
\end{theorem}
\vspace*{0.3cm}
 We discuss again the toss game from Example~\ref{ex:toss_game_1}.
\begin{example}\label{ex:toss_game_2} 
  ($2:1$ toss game; $M=3$)\\
  As before $N=2$, $p_i=\frac{1}{2}$, $t_1=-1$, $t_2=2$ and $\hat{t}=1$. The loss $t_1=-1$ will occur if the coin shows tail (T) and $t_2=2$ corresponds to head (H).
  Depending on $\ell^\ast\!=\ell^\ast(f,\omega)$ with $f>0$ sufficiently small we get the following trade series realizing their maximum with the $\ell^\ast$th toss.
  (cf. Definition~\ref{def:l*})
  \begin{description}
   \item[$\ell^\ast\!=3$:] $(H,H,H)$; $(H,T,H)$; $(T,H,H)$. Hence
     \begin{align*}
       R_n^{(3)} = R_n^{(3,M=3,N=2)}
       = \begin{cases}
           \frac{2}{8},&\text{for $n=1$}\\
           \frac{7}{8},&\text{for $n=2$}
         \end{cases}
      \qquad\text{and always }\Lambda_n^{(3)}=0.
     \end{align*}
   \item[$\ell^\ast\!=2$:] $(H,H,T)$; $(T,H,T)$. Hence
     \begin{align*}
       R_n^{(2)}
       = \begin{cases}
           \frac{1}{8},&\text{for $n=1$,}\\
           \frac{3}{8},&\text{for $n=2$,}
         \end{cases}
      \qquad\text{and }\qquad \Lambda_n^{(2)}
       = \begin{cases}
           \frac{2}{8},&\text{for $n=1$,}\\
           0,&\text{for $n=2$.}
         \end{cases}
     \end{align*}
   \item[$\ell^\ast\!=1$:] $(H,T,T)$. Hence
     \begin{align*}
       R_n^{(1)}
       = \begin{cases}
           0,&\text{for $n=1$,}\\
           \frac{1}{8},&\text{for $n=2$,}
         \end{cases}
      \qquad\text{and }\qquad \Lambda_n^{(1)}
       = \begin{cases}
           \frac{2}{8},&\text{for $n=1$,}\\
           0,&\text{for $n=2$.}
         \end{cases}
     \end{align*}
   \item[$\ell^\ast\!=0$:] $(T,T,T)$; $(T,T,H)$. Hence
     \begin{align*}
       R_n^{(0)} = 0
      \qquad\text{and }\qquad \Lambda_n^{(0)}
       = \begin{cases}
           \frac{5}{8},&\text{for $n=1$,}\\
           \frac{1}{8},&\text{for $n=2$.}
         \end{cases}
     \end{align*}
  \end{description}
  Therefore $\sum_{\ell=0}^{M=3}\Lambda_1^{(\ell)}=\frac{9}{8}$ and \\ $\sum_{\ell=0}^{M=3}\Lambda_2^{(\ell)}=\frac{1}{8}$ and Theorem~\ref{theo:EDD} yields
  \begin{align}\label{eq:E cur} 
    \E(\cDD^{(M=3)}(f,\cdot)) = \frac{9}{8}\log(1-f)+\frac{1}{8}\log(1+2f)
  \end{align}
  for all $f>0$ sufficiently small. Analogously from $\sum_{\ell=0}^{M=3}R_1^{(\ell)}=\frac{3}{8}$ and $\sum_{\ell=0}^{M=3}R_2^{(\ell)}=\frac{11}{8}$ and
  Theorem~\ref{theo:ER} we get
  \begin{align}\label{eq:cR log} 
    \E(\cR^{(M=3)}(f,\cdot)) = \frac{3}{8}\log(1-f)+\frac{11}{8}\log(1+2f)
  \end{align}
  for all $f>0$ sufficiently small.
\end{example}
\vspace*{0.2cm}
\begin{remark}\label{rem:E representation} 
  The representations of\ $\E(\cDD^{(M)}(f,\cdot))$ and $\E(\cR^{(M)}(f,\cdot))$ from Theorems~\ref{theo:EDD} and \ref{theo:ER} clearly hold true
  only for sufficiently small $f>0$. For $f>0$ no longer small \\ the formulas for these expectation values change since the topping position
  $\ell^\ast\!=\ell^\ast(f,\omega)$ changes. To see that, consider $\omega_0 = (H,T,H)$ for the 2:1 toss game from above, but assume now that $f \in (0,1)$
  is so large such that the gain of the last $H$ toss does not compensate the loss of the $T$ toss from the 2nd toss, i.e. in case
  \begin{align*}
    \big|\log\rk{1-f}\big|\,>\,\log\rk{1+2f}\,.
  \end{align*}

  For those $f\,\text{we get}\ \ell^\ast\!\rk{f,\omega_0}=1$, which immediately results in different formulas for the expectation values of run-up and current drawdown.
\end{remark}



  \begin{figure}[htb]
    \centering
    \begin{minipage}[c]{1\linewidth}
        \centering
       \includegraphics[width=1.0\textwidth]{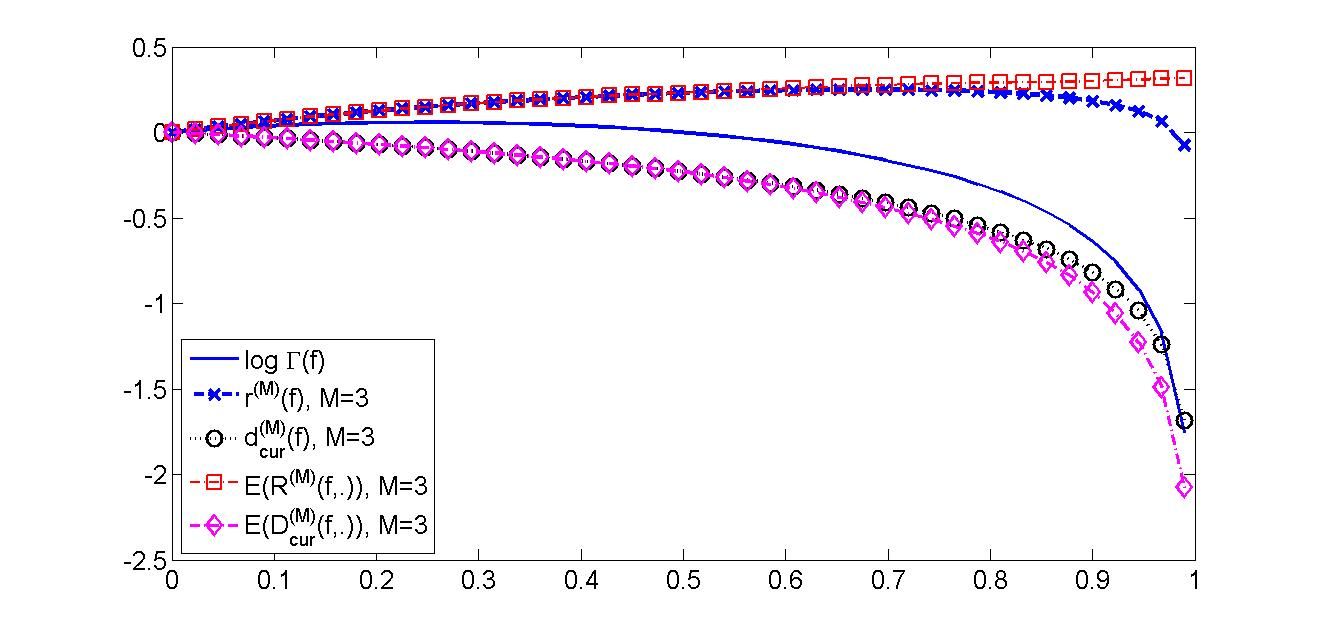}
    \end{minipage}
    \caption{    $\E(\cDD^{(3)}(f,\cdot))$ and $\E(\cR^{(3)}(f,\cdot))$ with their approximations of  \eqref{eq:EDD}  and  \eqref{eq:ER} }
    \label{figDcur}
  \end{figure}


 Again the approximations of current drawdown and run-up are quite accurate up to\ $f=0.6$ (see Figure \ref{figDcur}).



\vspace*{0.7cm}
   \section{Optimal $\mathbf{f}$ for risk averse fractional trading using the current drawdown}  \label{sec:6}

 Now we bring together the results of the previous sections. We saw in Theorem~\ref{theo:EZ} that the usual optimal $f$ problem which maximizes the 
 terminal wealth relative
\begin{align*}
  \TWR(f) = \prod_{i=1}^N\rk{1+f\frac{t_i}{\hat{t}}}\stackrel{\mbox{\large!}}{=}\max,\qquad \text{for }f\in[0,1)
\end{align*}
 is equivalent to maximizing
\begin{align*}
  \E(\cZ^{(M)}(f,\cdot)) = M\log\Gamma(f) = M\log\prod_{i=1}^N\rk{1+f\frac{t_i}{\hat{t}}}^{1/N}\stackrel{\mbox{\large!}}{=}\max,\qquad \text{for }f\in[0,1),
\end{align*}
 where $\cZ^{(M)}(f,\omega)=\log(\TWR_1^M(f,\omega))$ and $t_1,\ldots,t_N\in\R\backslash\{0\}$ are all trades occurring with the same probability $p_i=\frac{1}{N}$.
 By Corollary~\ref{cor:EU+ED} this is equivalent to maximize
\begin{align*}
  \E(\cU^{(M)}(f,\cdot))+\E(\cD^{(M)}(f,\cdot))\stackrel{\mbox{\large!}}{=}\max,\qquad \text{for }f\in[0,1).
\end{align*}

 This optimization problem clearly differentiates between \textbf{chance} ($\cU^{(M)}(f,\omega)\geq 0$) and \textbf{risk} ($\cD^{(M)}(f,\omega)\leq 0$) 
 parts. A drawdown averse investor may, however, not only take a look at the downtrade log series $\cD^{(M)}(f,\omega)$ but may as well look at the
 current drawdown $\cDD^{(M)}(f,\omega)$, because the current drawdown is in a way that part of the investment process in risky assets, which ``hurts'' 
 every day. Since
\begin{align}\label{eq:DD<D} 
  \cDD^{(M)}(f,\omega) \leq \cD^{(M)}(f,\omega) \le 0
\end{align}
 we propose
\begin{align}\label{eq:opt_with_DD} 
  \E(\cU^{(M)}(f,\cdot))+\E(\cDD^{(M)}(f,\cdot))\stackrel{\mbox{\large!}}{=}\max,\qquad \text{for}\quad f \in [0,1)
\end{align}
 as a more risk averse optimization problem. From the discussion in the sections before, it got clear that those $\omega\in\Omega^{(M)}$, which contribute non-trivial values to the calculation of the above two expectation values, do depend on $f$. Therefore \eqref{eq:opt_with_DD} might be too hard to solve
 in general at least for $M$ large. Nevertheless, the Theorems~\ref{theo:EU_smallf} and \ref{theo:EDD} give explicitly calculable  formulas for $\E(\cU^{(M)}(f,\cdot))$ and $\E(\cDD^{(M)}(f,\cdot))$ for all
 sufficiently small $f>0$. We therefore propose as alternative to maximize
\begin{align}\label{eq:opt_with_DD_approx} 
  \sum_{n=1}^N\underbrace{\ek{U_n^{(M,N)} + \sum_{\ell=0}^M\Lambda_n^{(\ell,M,N)}}}_{=:q_n^{(M,N)}=:q_n \ge 0}\cdot \log\rk{1+f\frac{t_n}{\hat{t}}}\stackrel{\mbox{\large!}}{=}\max,
  \qquad \text{for}\quad f \in [0,1)
\end{align}

 with the hope, that the $q_n$ yield for $f$ no longer small still good approximations for \eqref{eq:opt_with_DD}.
 Fortunately the problem \eqref{eq:opt_with_DD_approx} was ``solved'' already in Section~\ref{sec:3}.
\begin{corollary} 
  For a trading system as in Setup~\ref{setup:Z} with $N,M \in \N$ fixed the optimization problem \eqref{eq:opt_with_DD_approx} with $q_n=q_n^{(M,N)}$
  \begin{align} 
    \sum_{n=1}^Nq_n\cdot\log\rk{1+f\frac{t_n}{\hat{t}}}\stackrel{\textnormal{\large!}}{=}\max,\qquad \text{for}\quad f \in [0,1)
  \end{align}
  has a unique solution $f=f^{\textnormal{opt},\cDD^{(M)}}\!\in (0,1)$ if\ $\sum_{n=1}^Nq_nt_n>0$ and $f^{\textnormal{opt},\cDD^{(M)}}=0$ in case \\
  $\sum_{n=1}^Nq_nt_n\leq 0$.
\end{corollary}
\begin{proof}
  Set\ $q:=\sum_{n=1}^Nq_n>0$\ and\ $\tilde{p}_n:=\frac{q_n}{q}$. The claim follows from Theorem~\ref{theo:optimal_f} and Remark~\ref{rem:optimal_f}.
\end{proof} \vspace*{0.3cm}
\begin{remark}\label{rem:opt_problem} 
  Since the optimization problem \eqref{eq:opt_with_DD_approx} was derived as approximation of the optimization problem \eqref{eq:opt_with_DD} for $f > 0$
  small, it is reasonable that small solutions $f^{\textnormal{opt},\cDD^{(M)}}$ may be good approximations to solutions of \eqref{eq:opt_with_DD}
\end{remark}
 We want to make the difference clear with the toss game:
\begin{example}\label{ex:toss_game_3} 
  ($2:1$ toss game; $M=3$)\\
  Using $N=2$, $p_i=\frac{1}{2}$, $t_1=-1$, $t_2=2$ and $\hat{t}=1$ the usual optimal $f$ solves
  \begin{align*}
    \TWR(f) = (1-f)(1+2f)\stackrel{\textnormal{\large!}}{=}\max,\qquad \text{for }f\in[0,1).
  \end{align*}
  Since this is also the situation of the Kelly formula
  \begin{align}\label{eq:kellyV} 
    f^{\textnormal{opt,KellyV}}=p-\frac{q}{B}
  \end{align}
  for a game where the win $B$ occurs with probability $p$ and the loss $-1$ occurs with probability $q=1-p$, we use $B=2$ and $p=q=\frac{1}{2}$ to obtain
   \begin{align*}
    f^{\textnormal{opt}}=f^{\textnormal{opt,KellyV}}=\frac{1}{4}=25\%.
  \end{align*}
  From Example~\ref{ex:toss_game_1}, \eqref{eq:EU_smallf cdot} and Example \ref{ex:toss_game_2}, \eqref{eq:E cur} we already know
  \begin{align*}
    \E(\cU^{(3)}(f,\cdot)) &= \frac{3}{8}\log(1-f)+\frac{9}{8}\log(1+2f)\\
  \intertext{and}
    \E(\cDD^{(3)}(f,\cdot)) &= \frac{9}{8}\log(1-f)+\frac{1}{8}\log(1+2f).
  \end{align*}
  Hence \eqref{eq:opt_with_DD_approx} is equivalent to
  \begin{align*}
    \frac{12}{8}\log(1-f)+\frac{10}{8}\log(1+2f)\stackrel{\textnormal{\large!}}{=}\max
    \quad\Leftrightarrow\quad
    \frac{6}{11}\log(1-f)+\frac{5}{11}\log(1+2f)\stackrel{\textnormal{\large!}}{=}\max
  \end{align*}
  and again with the Kelly formula \eqref{eq:kellyV} with $p=\frac{5}{11}$ and $q=\frac{6}{11}$ we get $f^{\textnormal{opt},\cDD^{(3)}}=\frac{2}{11}\approx 18\%$.
\end{example}



  \begin{figure}[htb]
    \centering
    \begin{minipage}[c]{1\linewidth}
        \centering
       \includegraphics[width=1.0\textwidth]{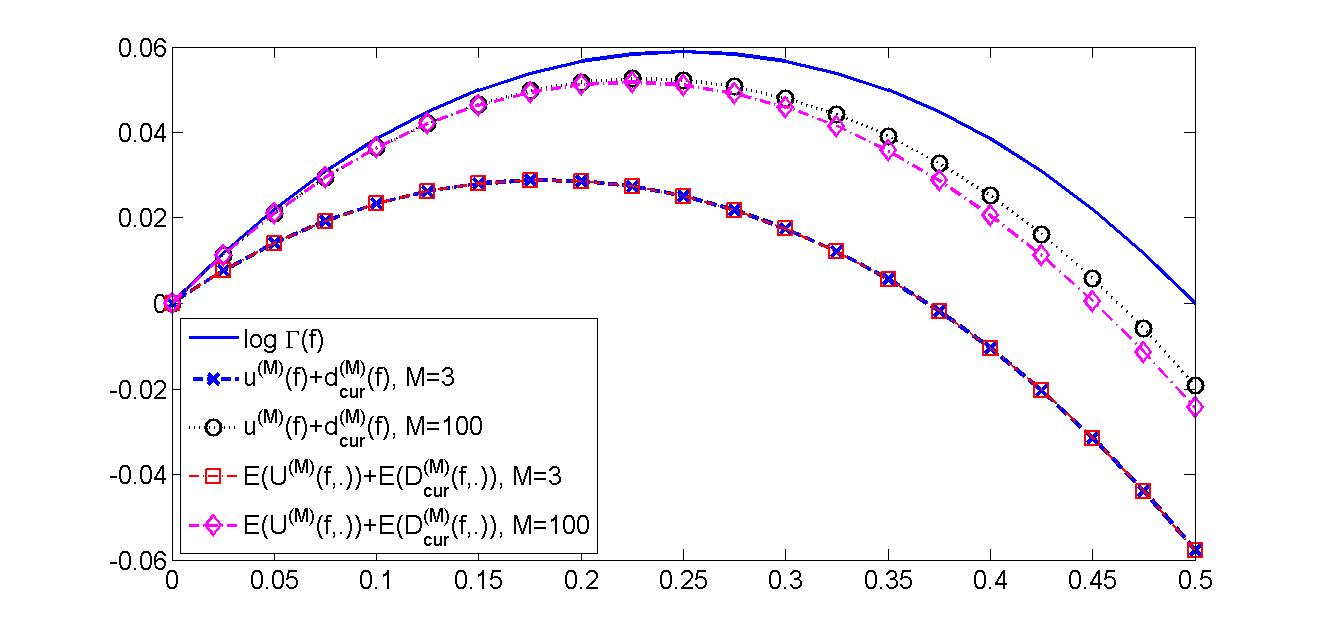}
    \end{minipage}
    \caption{$\E(\cU^{(M)}(f,\cdot))+\E(\cDD^{(M)}(f,\cdot))$ for $M=3$ and $M=100$ including
             their approximations}
    \label{figlog-gamma}
  \end{figure}


 In Figure \ref{figlog-gamma} we can see that the optimization problems \eqref{eq:opt_with_DD} and \eqref{eq:opt_with_DD_approx}
 for\ $M=3$\ are completely equivalent and even for\ $M=100$\ the approximated problem comes very close. Therefore the solutions
 of \eqref{eq:opt_with_DD} and \eqref{eq:opt_with_DD_approx} should be close too.



\begin{table}[htb]
 \begin{displaymath}\hspace{-0.1cm}
 \begin{tabular}{|@{\hspace{0.1cm}}l@{\hspace{0.1cm}}||@{\hspace{0.1cm}}c@{\hspace{0.1cm}}|@{\hspace{0.1cm}}c@{\hspace{0.1cm}}|@{\hspace{0.1cm}}c@{\hspace{0.1cm}}|
                  @{\hspace{0.1cm}}c@{\hspace{0.1cm}}|@{\hspace{0.1cm}}c@{\hspace{0.1cm}}|@{\hspace{0.1cm}}c@{\hspace{0.1cm}}|@{\hspace{0.1cm}}c@{\hspace{0.1cm}}|
                  @{\hspace{0.1cm}}c@{\hspace{0.1cm}}|@{\hspace{0.1cm}}c@{\hspace{0.1cm}}|@{\hspace{0.1cm}}c@{\hspace{0.1cm}}| }                                                                 \hline
                                                   &            &            &            &            &            &            &            &            &            &             \\ [-0.3cm]
   $\;\mathbf{M}$                                  &\textbf{2}  &\textbf{3}  &\textbf{4}  &\textbf{5}  &\textbf{6}  &\textbf{7}  &\textbf{8}  &\textbf{9}  &\textbf{10} &\textbf{15}  \\ [0.12cm]\hline
                                                   &            &            &            &            &            &            &            &            &            &             \\ [-0.3cm]
   $\mathbf{f^{\textbf{opt},\mathbf{\cDD^{(M)}}}}$ & 0,1667     & 0,1818     & 0,1739     & 0,1613     & 0,1839     & 0,1758     & 0,1685     & 0,1870     & 0,1802     & 0,1926      \\ [0.12cm] \hline\hline\hline
                                                   &            &            &            &            &            &            &            &            &            &             \\ [-0.3cm]
   $\;\mathbf{M}$                                  &\textbf{20} &\textbf{25} &\textbf{30} &\textbf{40} &\textbf{50} &\textbf{60} &\textbf{70} &\textbf{80} &\textbf{90} &\textbf{100} \\ [0.12cm]\hline
                                                   &            &            &            &            &            &            &            &            &            &             \\ [-0.3cm]
   $\mathbf{f^{\textbf{opt},\mathbf{\cDD^{(M)}}}}$ & 0,1898     & 0,1980     & 0,2043     & 0,2094     & 0,2145     & 0,2197     & 0,2229     & 0,2258     & 0,2283     & 0,2302      \\ [0.12cm]\hline
    \end{tabular}
 \end{displaymath}
 \centering
 \caption{Optimal fraction for the risk aware optimization problem \eqref{eq:opt_with_DD_approx} in the 2:1 toss game from Example~\ref{ex:toss_game_3}.
 \label{tab1} }
\end{table}


 In Table \ref{tab1} we see the optimal solution $f^{\textnormal{opt},\cDD^{(M)}}$ of \eqref{eq:opt_with_DD_approx} for the 2:1 toss game and for a selected 
 set of $M$ values. It seems that, as $M$ increases, the optimal solutions approach the optimal Kelly fraction\ $f^{\textnormal{opt,KellyV}}\!= 25\%$.\
 To invest more risk averse it therefore would be natural to use the minimum of the optimal solutions from Table \ref{tab1}, which is close to\
 $16\% =: f^{\textnormal{opt,cur}\;\cD\!\cD}$. \vspace*{0.5cm}

 In the remainder of this section we would like to give a simulation of the 2:1 toss game to see the difference of the above mentioned two fractions.
 Each of the following simulations uses a starting capital of $1000$ and draws $10.000$ instances of the 2:1 toss game independently. In Figure \ref{fig-equity}
 we see the resulting $\log$--equity curves for\ $n=1,\ldots,10.000$ in black and as a reference the expected $\log$--equity lines dotted

\vspace*{0.3cm}

  \begin{figure}[htb]
    \centering
    \begin{minipage}[c]{0.5\linewidth}
        \centering
        \includegraphics[width=0.9\textwidth]{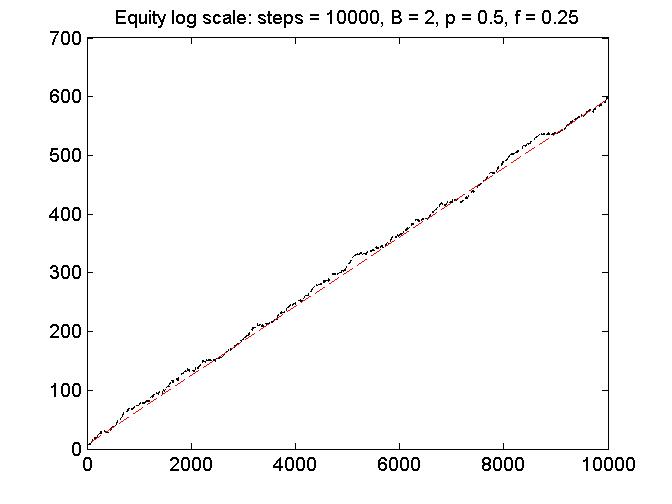}
    \end{minipage}
\hspace{-0.6cm}
    \begin{minipage}[c]{0.5\linewidth}
        \centering
        \includegraphics[width=0.9\textwidth]{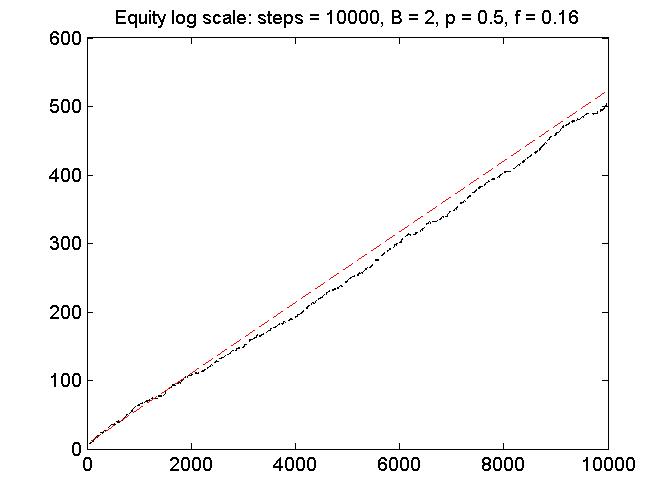}
    \end{minipage}
    \caption{\textit{Log}--equity curve for the 2:1 toss game with
             $f^{\textnormal{opt,KellyV}}\!\!=\!25\%$ vs $f^{\textnormal{opt,cur}\,\cD\!\cD}\!=\!16\%$}
  \label{fig-equity}
  \end{figure}

\newpage

 Clearly the wealth growth according to $f^{\textnormal{opt,cur}\,\cD\!\cD}\!=16\%$ is less than the wealth much growth of $f^{\textnormal{opt,KellyV}}\!= 25\%$,
 but the reduction is reasonable. The question remains how better is the risk side for the risk aware strategy. In the Figure \ref{fig-drawdown} we see
 a plot of the current relative drawdown (negative) displayed as a so called ``blood curve''.

\vspace*{0.3cm}

  \begin{figure}[htb]
    \centering
    \begin{minipage}[c]{0.5\linewidth}
        \centering
        \includegraphics[width=0.9\textwidth]{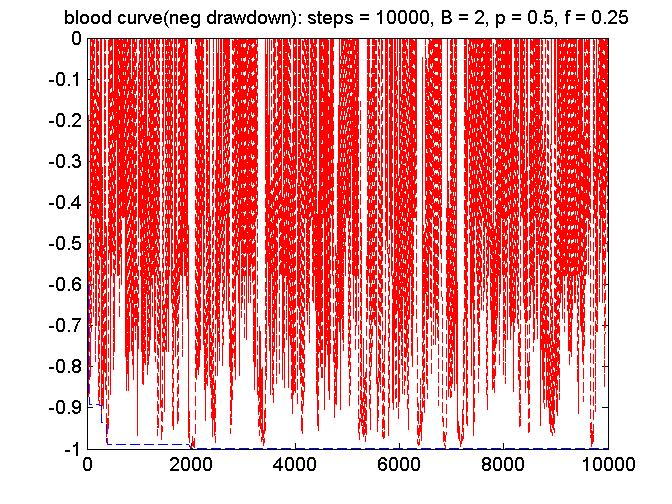}
    \end{minipage}
\hspace{-0.6cm}
    \begin{minipage}[c]{0.5\linewidth}
        \centering
        \includegraphics[width=0.9\textwidth]{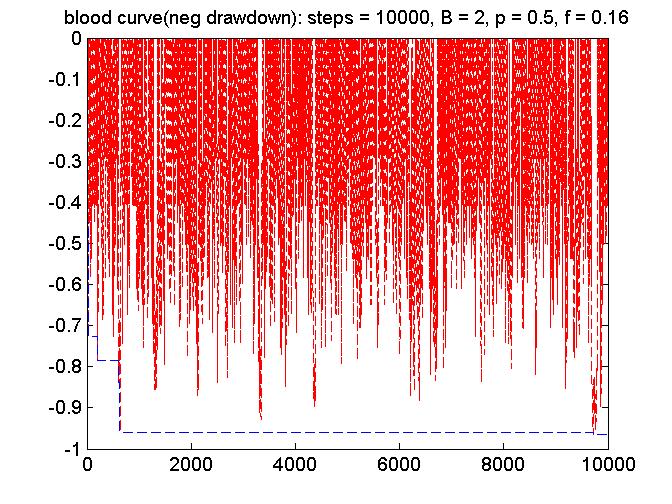}
    \end{minipage}
    \caption{Current relative drawdown (negative) for the 2:1 toss game with
            $f^{\textnormal{opt,KellyV}}\!\!=\!25\%$ vs $f^{\textnormal{opt,cur}\,\cD\!\cD}\!=\!16\%$}
  \label{fig-drawdown}
  \end{figure}

\vspace*{0.5cm}

 One can see that the maximal relative drawdown for $f^{\textnormal{opt,KellyV}}$ lies around $-99\%$ whereas for $f^{\textnormal{opt,cur}\,\cD\!\cD}$
 it comes close to $-95\%$ for this simulation. More importantly, relative drawdowns of more than $-80\%$ become rare events for the risk
 averse strategy which was not the case for the Kelly optimal $f$ strategy. Looking at the distribution of the relative drawdowns
 (see Figure \ref{fig-distribution}) this will became explicit.

\vspace*{0.3cm}

  \begin{figure}[htb]
    \centering
    \begin{minipage}[c]{0.5\linewidth}
        \centering
        \includegraphics[width=0.9\textwidth]{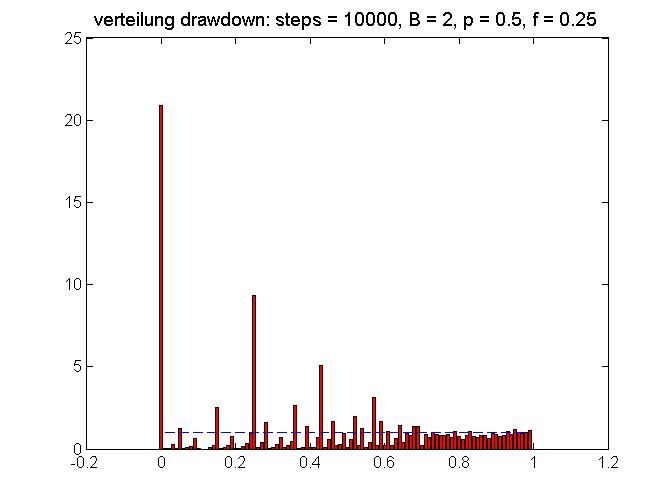}
    \end{minipage}
\hspace{-0.6cm}
    \begin{minipage}[c]{0.5\linewidth}
        \centering
        \includegraphics[width=0.9\textwidth]{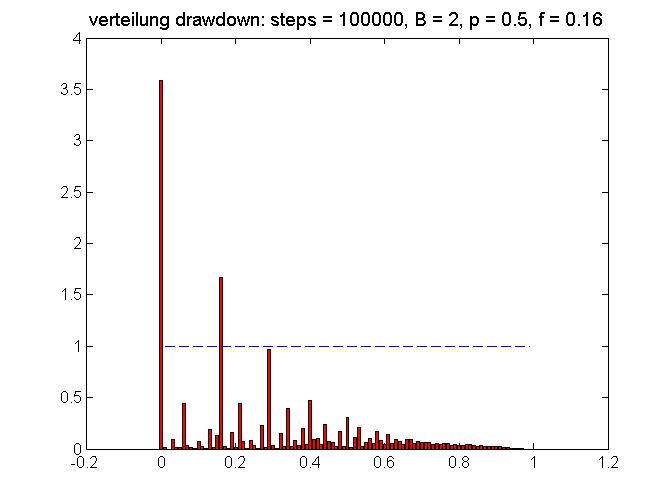}
    \end{minipage}
    \caption{Distribution of the current relative drawdown (positive) for the 2:1 toss game with
             $f^{\textnormal{opt,KellyV}}\!\!=\!25\%$ vs $f^{\textnormal{opt,cur}\,\cD\!\cD}\!=\!16\%$}
  \label{fig-distribution}
  \end{figure}


\newpage 


   \section{Conclusion}  \label{sec:7}

 The splitting of the goal function of fractional trading into ``risk'' and ``chance'' parts made it possible to introduce new
 more risk aware goal functions. This is carried out using the current drawdown and results in more defensive money management
 strategies. However, as simulations in section \ref{sec:6} show, even the new risk averse strategy might still be too risky for
 investing with real money. One alternative might be to use the maximal drawdown (on $M$ trades) instead of the current drawdown
 in the risk averse optimization problem \eqref{eq:opt_with_DD}. Therefore a similar result as Theorem \ref{theo:EDD} would be
 desirable for the maximal drawdown as well. Whether or not that is possible remains an open question.                 \vspace*{0.3cm}                    

 So far these strategies only work for single asset portfolios. The theory of fractional trading of portfolios was introduced
 by Vince \cite{vince:rpm09} with his leverage space trading model. Furthermore, Hermes \cite{hermes:mft2016} has extended
 the portfolio theory of fractional trading to trading results with continuous distributions. Nevertheless, also the question
 how risk averse strategies may be used for portfolios with many different assets to be traded simultaneously is still open.



\vspace*{0.8cm}


\end{document}